\documentclass[reqno]{amsart}
\usepackage{amsaddr} 
\usepackage{amsmath,amssymb}
\usepackage{nccmath}  
\usepackage{graphics,epsfig,color}
\usepackage[mathscr]{euscript}
\usepackage{enumerate}      
\usepackage{dsfont}         
\usepackage{color}

\usepackage{tikz}           
\usepackage{subfig,caption}         

\usepackage[british]{babel}

\usepackage{hyperref}

\theoremstyle{plain}
\newtheorem{theorem}{Theorem}[section]

\newtheorem{proposition}[theorem]{Proposition}

\newtheorem{corollary}[theorem]{Corollary}

\theoremstyle{definition}
\newtheorem{definition}[theorem]{Definition}

\theoremstyle{remark}
\newtheorem{remark}[theorem]{Remark}

\numberwithin{equation}{section}
\numberwithin{theorem}{section}

\newcommand{\upbar}[1]{\,\overline{\! #1}}

\renewcommand{\epsilon}{\varepsilon}
\renewcommand{\tilde}{\widetilde}
\renewcommand{\hat}{\widehat}
\newcommand{\A}{\mathcal{A}}
\newcommand{\BV}{\mathrm{BV}}

\renewcommand{\div}{\mathop{\rm div}\nolimits}

\newcommand{\M}{\mathcal{M}}
\newcommand{\NN}{\mathbb{N}}

\newcommand{\PP}{\mathbb{P}}
\newcommand{\I}{\mathcal{I}}
\renewcommand{\L}{\mathcal{L}}
\newcommand{\RR}{\mathbb{R}}

\newcommand{\X}{\mathcal{X}}

\newcommand{\TW}{\mathrm{TW}}
\newcommand{\const}{\mathrm{const}}

\DeclareMathOperator\co{co}
\DeclareMathOperator\Ran{Ran}

\DeclareMathOperator\Prob{Prob}

\definecolor{light}{gray}{.9}

\title[Dynamical phase transitions on finite graphs]{Dynamical phase transitions for flows on finite graphs}

\author{Davide Gabrielli}
\address{\small{Universit{\`a} dell'Aquila\\ Via Vetoio, Loc. Coppito\\ 67010 L'Aquila, Italia.}}
\author{D.R. Michiel Renger}
\address{\small{WIAS Berlin\\Mohrenstrasse 39\\10117 Berlin, Germany}}

\begin{document}

\maketitle
\thispagestyle{empty}

\begin{abstract}
We study the time-averaged flow in a model of particles that randomly hop on a finite directed graph. In the limit as the number of particles and the time window go to infinity but the graph remains finite, the large-deviation rate functional of the average flow is given by a variational formulation involving paths of the density and flow. We give sufficient conditions under which the large deviations of a given time averaged flow is determined by paths that are constant in time. We then consider a class of models on a discrete ring for which it is possible to show that a better strategy is obtained producing a time-dependent path. This phenomenon, called a dynamical phase transition, is known to occur for some particle systems in the hydrodynamic scaling limit, which is thus extended to the setting of a finite graph.
\bigskip

\noindent{\em Keywords}: Large deviations, particle systems, phase transitions.

\noindent{\em AMS 2010 Subject Classification}:
60F10, 05C21, 82C22, 82C26  
\end{abstract}

\section{Introduction}
\label{sec:introduction}

One of the main challenges of statistical mechanics is to understand the thermodynamics of particle systems that are not in detailed balance. A violation of detailed balance means that even in a steady state, there can be non-trivial net flows. It is therefore natural to study flows, currents and their corresponding large deviations, which is the basis of macroscopic fluctuation theory~\cite{BDSGJLL2015MFT}.
Various limits and large deviations of flows and currents have been studied in the literature, e.g. steady-state or pathwise large deviations as the number $N$ of particles and/or lattice sites goes to infinity, e.g. \cite{Derrida2007,BDSGJLL2015MFT}, or large deviations of time-averaged flows as time $T$ goes to infinity, e.g.~\cite{BertiniFaggionatoGabrielli2016Flows,BertiniFaggionatoGabrielli2016RigLDP}. In this work we are concerned with large deviations of time-averaged flows $\tfrac1T\int_0^T\!Q^N(t)\,dt$, where \emph{first} the number $N$ of particles and \emph{then} the time horizon $T$ is sent to infinity, as for example in~\cite{MR2227084,BDSGJLL2005CurFluct,BDSGJLL2007LdpEmpCur,BodineauDerrida2004,BodineauDerrida2005}. This yields a large deviation principle of the type:
\begin{align*}
  \Prob\!\left( {\textstyle \tfrac1T\int_0^T\!Q^N(t)\,dt \approx \overline q} \right) \stackrel{\substack{N\to\infty\\T\to\infty}}{\sim} e^{-NT\Psi(\overline q)},\notag\\
\intertext{with}
  \Psi(\overline q):=\lim_{T\to\infty} \inf_{\substack{(\rho,q):\\ \tfrac1T\int_0^T\!q(t)\,dt=\overline q,\\ \dot\rho(t)+\div q(t)=0}} \mfrac{1}{T}\int_0^T\!\L(\rho(t),q(t))\,dt,
\end{align*}
where the infimum ranges over paths of particle densities $\rho(t)$ and particle flows $q(t)$, and $\L$ is some non-negative cost function.

In many cases, the infimum over paths is attained or approximated by constant paths, i.e. $\rho(t)$ constant and $q(t)$ constant and divergence free. In this case the above expression simplifies to $\Psi(\overline q)=\inf_{\rho}\L(\rho,\overline q)$ ~\cite{BodineauDerrida2004}. Such simplification will fail if a time-dependent flow $q(t)$ is significantly less costly than its corresponding time-averaged flow $\overline q$. In that case we say that a \emph{dynamical phase transition} occurs (the transition being in parameter $\overline q$), see Definition~\ref{def:DFT}.

Dynamical phase transitions have been shown to occur for systems like exclusion processes with an external bulk force~\cite{BodineauDerrida2005} or the KMP model~\cite{MR2227084}, in the hydrodynamic scaling limit. See also \cite{PhysRevLett.107.180601} for numerial simulations and \cite{PhysRevLett.118.115702}, \cite{Garrahan_2009} and references therein for problems of this type for different models. In this case the mesh of the underlying graph goes to zero in the many-particle limit, yielding quadratic large-deviation cost functions, i.e. $\L(\rho,j)=\tfrac12\lVert j - k(\rho)\rVert^2_{L^2(1/\chi(\rho))}$ for some mobility $\chi(\rho)$ and expected flow $k(\rho)$. By contrast, we consider a finite graph $(\X,E)$ that is not rescaled but remains discrete in the limit; in that case the cost function is typically \emph{entropic}, i.e. $\L(\rho,q)=\sum_{(x,y)\in E} s\big(q_{x,y}\mid k_{x,y}(\rho)\big)$, where $s$ is a relative entropy like functional of $q_{x,y}$ with respect to $k_{x,y}(\rho)$. The precise model, assumptions, and large-deviations statements are introduced in Section~\ref{sec:particle system}.

Entropic cost functions are in many aspects more challenging than quadratic ones. For example, they can only induce generalised gradient flows~\cite{MielkePeletierRenger2014}, they can only be decomposed using a generalised notion of orthogonality~\cite{KaiserJackZimmer2018,RengerZimmer2019TR}, and it is not clear whether they relate to some manifold or more general geometry.

In this work we focus on jump processes on finite graphs with corresponding entropic cost functions, and study whether and when dynamical phase transitions occur. 

The paper is organised as follows. In Section~\ref{sec:particle system} we introduce the general setting that we work in; in Section~\ref{sec:convexity} we present general properties of dynamical phase transitions, and in particular sufficient assumptions to rule them out. In Section~\ref{sec:discrete ring} we restrict to zero-range processes on a discrete ring and give sufficient conditions under which a dynamical phase transition does occur.


\section{Particle systems, limits and large deviations}
\label{sec:particle system}

In Section~\ref{subsec:particle system} we introduce the general particle systems that we work with; the limiting behavior and large deviations as $N\to\infty$ are discussed in Section~\ref{subsec:limit and ldp N},  and the limit and large deviations of average flows as $N\to\infty,T\to\infty$ are discussed in Section~\ref{subsec:ldp time-average}.

\subsection{Particle systems}
\label{subsec:particle system}

Let $(\X, E)$ be a finite directed graph, where the directed edges are $E\subseteq \X\times\X$, excluding the diagonal. The particle system consists of $N$ identical particles $X_1(t),\hdots,X_N(t)$ on $\X$.

The initial states $(x_i)_{i=1}^N$ of the particles are chosen deterministically so that the empirical measure converges, as $N\to +\infty$, to some fixed measure:
\begin{equation*}
  \rho^N_x(0):=\frac mN\sum_{i=1}^N\mathds1_{x_i}(x) = \frac mN\#\{\text{particles at site } x \text{ at time } 0\} \to \rho^*_x.
\end{equation*}
Each single particle carries a mass of $m/N$, so that the total mass in the system can be tuned by the parameter $m>0$, and $\rho^N(0),\rho^*\in\M_m(\X)$, the set of non-negative measures on $\X$ of total mass $m$. The precise choice of the initial condition does not play a key role in this paper, since from the next section on we are only interested in the long-time behaviour.

For the dynamics of the system, one particle randomly hops from site $x$ to $y$ with model-specific rate $N k_{x,y}(\rho^N(t))$ where the \emph{empirical measure} is:
\begin{equation}
  \rho^N_x(t):=\frac mN\sum_{i=1}^N\mathds1_{X_i(t)}(x) = \frac mN\#\{\text{particles at site } x \text{ and time } t\}.
\label{eq:discr emp measure}
\end{equation}
The empirical measure $\rho^N(t)$ is a Markov process in $\M_m(\X)$ with generator
\begin{equation}
  \hat L^N f(\rho) = \sum_{(x,y)\in E}\, Nk_{x,y}(\rho) \big\lbrack f(\rho + \tfrac mN(\mathds1_y -\mathds1_{x}) - f(\rho)\big\rbrack.
\label{eq:discr dens generator}
\end{equation}

For this model, we introduce the \emph{integrated flow} on an edge $(x,y)\in E$  as
\begin{equation}
  W^N_{x,y}(t):=\frac mN \#\big\{\text{particles jumped from $x$ to $y$ in time interval } (0,t\rbrack \big\}\,.
\label{eq:discr emp flow}
\end{equation}
This definition also fixes the initial condition $W^N(0)\equiv0$. The integrated flow $W^N(t)\in\M(E)$ is a non-negative measure on edges, related to the empirical measure through the discrete continuity equation:
\begin{align}\label{eq:discr cont eq}
  \rho^N(t) = \rho^N(0) - \div W^N(t)\, &&\text{where}&&
  \div W_x:= \sum_{y:(x,y)\in E} W_{x,y} - \sum_{y:(y,x)\in E}W_{y,x}.
\end{align}

The pair $(\rho^N(t),W^N(t))$ is again a Markov process, now in $\M_m(\X)\times \M(E)$, with generator
\begin{equation*}
  L^Nf(\rho,w) = \sum_{(x,y)\in E}\, Nk_{x,y}(\rho)\left[f\Big(\rho +\tfrac mN\left(\mathds1_{y}-\mathds1_{x}\right), w + \tfrac mN\mathds1_{(x,y)}\Big) - f\Big(\rho,w\Big)\right].
\end{equation*}
The initial condition will be always of the form $W^N(0)=0$ and $\rho^N_x(0)=\frac{n_xm}{N}$ for some integers $n_x$ such that $\sum_{x\in \X}n_x=N$.
Since the intensities depend on the density $\rho$ only, ignoring the integrated flows reduces this process to the process with generator~\eqref{eq:discr dens generator}.



Note that the rates $k_{x,y}(\rho)$ may depend on all coordinates $(\rho_z)_{z\in\X}$, allowing for example for attraction or catalysis. Here we recall the technical assumtions on $k$ used in \cite{PattersonRenger2019} to prove a large deviations principle. The exact assumptions on $k$ do not play a key role in the current paper, but we briefly mention them here for completeness:
\begin{enumerate}
\item each $k_{x,y}(\rho)\to0$ sufficiently fast when $\rho_x\to0$;
\item each $k_{x,y}(\rho)$ is Lipschitz continuous and non-decreasing in $\rho$ with respect to the natural partial ordering on $\M_m(\X)$;
\item there exists a non-decreasing bijection $\phi:\lbrack0,1\rbrack\to\lbrack0,1\rbrack$ so that
$k_{x,y}(\delta \rho)\geq \phi(\delta) k_{x,y}(\rho)$
for all $\rho\in\M_m(\X),\delta\in\lbrack0,1\rbrack$ and $(x,y)\in E$.
\end{enumerate}
The first condition is needed to prevent negative densities. The continuity implies in particular that the process does not blow up and that the random paths are almost surely of bounded variation, so that the time derivatives $(\dot\rho^N(dt),\dot W^N(dt))$ always exist as measures in time. We call $Q^N(dt):=\dot W^N(dt)$. The monotonicity and the lower bound are rather technical assumptions needed in the proof of the large-deviation principle Theorem~\ref{th:discr path ldp} that we discuss in the next section.

\begin{remark}
In the limit $N\to\infty$ and the corresponding large-deviation regime, the paths become more regular and their derivatives exists as $L^1$-densities $(\dot\rho(t),\dot w(t))$. The precise topology on the path space and rigorous proofs of scaling limits and large-deviations principles will not be discussed here. We will use the notation $q(t):=\dot w(t)$.
\label{rem:BV}
\end{remark}

\begin{remark}
We believe that our results can be extended to more general models. For example, the microscopic jump rates can be $Nk_{x,y}(\rho)$ in approximation only, and the continuity equation~\eqref{eq:discr cont eq} could include more general linear operators than the discrete divergence. This would include the classical Kurtz model for chemical reactions~\cite{Kurtz1970}, allowing also for annihilation and creation of mass. We could consider even models with mass randomly flowing  across each edge and not necessarily transported by particles, like for example in the KMP model \cite{MR656869}. The monotonicity of the rates $k_{x,y}$ should also be generalisable as to include repulsion effects, but the best of our knowledge this has not yet been carried out in the literature.

\end{remark}

\subsection{Limit and large deviations as $N\to\infty$}
\label{subsec:limit and ldp N}


The many-particle limit follows by a similar argument as the classic results by Kurtz:
\begin{theorem}[{\cite{Kurtz1970,Kurtz1972,PattersonRenger2019}}]
The random paths $(\rho^N, W^N)\in\BV(0,T;\M_m(\X)\times\M(E)$ converge narrowly as $N\to\infty$ to the solution $(\rho,w)\in W^{1,1}(0,T;\M_m(\X)\times\M(E))$ of the problem
\begin{align}
&\dot w(t) = k(\rho(t)),
&
&\dot\rho(t) + \div\dot w(t)=0, \notag\\
&w(0)=0,
&
&\rho(0)=\rho^*.
\label{eq:ODE}
\end{align}
\end{theorem}

Fluctuations around this limit are quantified by the following large-deviation principle. Recall that we use the notation $w(t)=:\int_0^t\!q(s)ds$. 
%

\begin{theorem}[{\cite{PattersonRenger2019}}]
The random paths $(\rho^N, W^N)\in\BV(0,T;\M_m(\X)\times\M(E))$ satisfy a large-deviation principle as $N\to\infty$ with speed $N$:
\begin{equation*}
  \Prob\big( (\rho^N, W^N) \approx (\rho,w)\big) \sim e^{-N \I_{\lbrack 0,T\rbrack}(\rho,\dot w\mid\rho^*)}
\end{equation*}
with rate functional
\begin{multline}
\I_{\lbrack0,T\rbrack}(\rho,q \mid \rho^*)=\\
\begin{cases}
	\sum_{(x,y)\in E}\int_0^T\!s\big( q_{x,y}(t) \mid k_{x,y}(\rho(t)) \big)\,dt,
            & \textrm{if } \rho\in W^{1,1}(0,T;\M_m(\X)), q \in L^1(0,T;\M(E)),\\
            & \quad \dot\rho+\div q=0 \text{ and } \rho(0)=\rho^*\,,\\
	+\infty,  & \textrm{otherwise}\,,
\end{cases}
\label{eq:LDP N}
\end{multline}
	where
\begin{align*}
s(a\mid b):=
\begin{cases}
a\log\mfrac{a}{b} -a +b,  &a,b>0,\\
b,                        &a=0,\\
+\infty,                   &\text{otherwise}.
\end{cases}
\end{align*}
\label{th:discr path ldp}
\end{theorem}
The reader is referred to \cite{Renger2018a,PattersonRenger2019} for a proof of the above Theorem. Note that the integrand in~\eqref{eq:LDP N} can be interpreted as a kind of Lagrangian with corresponding Hamiltonian $\mathcal H(\rho,p):=\sum_{(x,y)\in E} k_{x,y}(\rho)\big(e^{p_{x,y}}-1\big)$, where $p$ are the variables dual to the variables $q$ \cite{Monthus2019}. See also \cite{MR3851785} for a similar rate functional obtained in the framework of time periodic Markov chains and  Appendix \ref{app:ldp proof} for the outline of a simple derivation in the case of independent particles using a Sanov Theorem for paths.

\begin{remark}
Clearly, the limit densities~\eqref{eq:ODE} satisfy the closed equation:
\begin{equation}
  \dot \rho(t)+\div k(\rho(t))=0\,.
\label{eq:closedrho}
\end{equation} 
In general it is possible to construct models for which the equations \eqref{eq:closedrho} have several stable and unstable equilibrium points or even limit cycles (see for example \cite{lazar}). We are however mostly interested in dynamical phase transitions for models for which \eqref{eq:closedrho} has a unique and globally attractive equilibrium. This corresponds to the behavior of a thermodynamic system relaxing to equilibrium.
\label{rem:single ss}
\end{remark}

\subsection{Limit and large deviations as $N\to\infty$ and $T\to\infty$}
\label{subsec:ldp time-average}

We now study the time-averaged flow (recall the notation introduced just before Remark~\ref{rem:BV}),
\begin{equation*}
  \overline Q^N(T):=\mfrac1T\int_0^T\! Q^N(dt):=\mfrac1T W^N(T),
\end{equation*}
in the limit where $N\to\infty$ while $T=T_N$ for some function of $N$ such that $T_N\to\infty$ sufficiently slow. To simplify the arguments we always take $N\to\infty$ first and subsequently $T\to\infty$. Under sufficient conditions, the average flow $\overline Q^N(T)$ converges to $k(\overline\rho)$ where $\overline\rho$ is the unique attractive steady state of~\eqref{eq:closedrho}.

In order to state the corresponding large-deviation principle we first need to introduce some notation and facts. Let 
\begin{equation*}
  \I_{\lbrack0,T\rbrack}(\rho,q):=\I_{\lbrack0,T\rbrack}(\rho,q\mid\rho(0))
\end{equation*}
be the functional that coincides with \eqref{eq:LDP N} but without the constraint on the initial condition. We then define
\begin{align}
  \Psi_{T,m}(\bar q\mid\rho^*)&:=\frac 1T\inf_{(\rho,q)\in \A_{T,m}(\bar q)} \I_{\lbrack0,T\rbrack}(\rho,q\mid\rho^*)\,,
\label{eq:PsiT}\\
  \Psi_{T,m}(\bar q)&:=\frac 1T\inf_{(\rho,q)\in \A_{T,m}(\bar q)} \I_{\lbrack0,T\rbrack}(\rho,q)\,,
\notag
\end{align}
where the infima are taken over the set,
\begin{multline}
\mathcal A_{T,m}(\overline q):= \Big\{(\rho, q)\in W^{1,1}(\lbrack 0,T\rbrack;\M_m(\X))\times L^1(\lbrack 0,T\rbrack;\M(E)) \;:\; \\
\dot\rho+\div q=0,\, \mfrac 1T\int_0^Tq(t)\,dt=\overline q \Big\}.
\label{eq:defA}
\end{multline}
Note there is a slight ambiguity in the notation~\eqref{eq:PsiT} since $\rho^*$ fixes the total mass $m$, but we include the index $m$ for consistency with the other function.
The following auxiliary proposition allows to show that the functions
\begin{align*}
  \Psi_m(\bar q\mid\rho^*):=\lim_{T\to +\infty}\Psi_{T,m}(\bar q\mid\rho^*) &&\text{and}&&
  \Psi_m(\bar q):=\lim_{T\to +\infty}\Psi_{T,m}(\bar q)
\end{align*}
are well defined and have certain properties. In particular, in this limit for large times, the dependence on the initial condition vanishes. 

\begin{proposition}[Well-posedness of $\Psi$] For any $\overline q\in\M(E)$ and $\rho^*\in \mathcal M_m(\X)$;
	\begin{enumerate}[(i)]
		\item $\Psi_m(\overline q\mid\rho^*)=\inf_{T>0} \Psi_{T,m}(\overline q\mid\rho^*)$ and $\Psi_m(\overline q)=\inf_{T>0} \Psi_{T,m}(\overline q)$ and hence the limits are well-defined;
		\item $\Psi_m(\overline q|\rho^*)$ and $\Psi_m(\overline q)$ are convex in $\overline q$;
		\item $\Psi_m(\overline q)=\infty$ if $\div\bar q\neq0$;
		\item $\Psi_m(\overline q\mid\rho^*)=\Psi_m(\overline q)$ \,\,(i.e. $\Psi_m$ does not depend on $\rho^*$); 
		\item the convergence in item $(i)$ is a pointwise and a Gamma-convergence. \label{it:Gamma-convergence}
		\end{enumerate}
	\label{prop:well-posedness of Psi}
\end{proposition}

\begin{proof}[Sketch of proof]
The proof follows the same arguments of \cite{BDSGJLL2005CurFluct, MR2227084, BDSGJLL2007LdpEmpCur}, and we do not discuss the full details. The subadditivity that implies existence and item $(i)$ is obtained by glueing paths together one after the other. Convexity follows by a similar construction.  If $\div \bar q\neq 0$ we obtain after a long enough time a negative mass somewhere that is impossible. The independence of the initial condition follows by the fact that one can construct a finite-time finite-cost path between any two initial conditions. The fact that the convergence is also a Gamma-convergence is obtained by combining item $(i)$ and $(ii)$ along the lines of \cite[Sec.~4]{BDSGJLL2007LdpEmpCur}.
\end{proof}

We can now state the main statement of this section.
\begin{theorem}\quad
\begin{enumerate}[(i)]
\item The average flow $\overline Q^N(T)$ satisfies a large-deviation principle with speed $N$ and rate functional 
$$
\bar q \to T\Psi_{T,m}(\bar q\mid \rho^*).
$$
\item The average flow $\overline Q^N(T)$ satisfies a large-deviation principle (where first $N\to\infty$ and then $T\to\infty$) with speed $NT$ and rate functional $\Psi_m(\bar q)$.
\end{enumerate}
\label{th:average flux LDP}
\end{theorem}

\begin{proof}[Sketch of proof] 
The first statement follows from a straight-forward contraction principle~\cite[Th.~4.2.1]{Dembo1998}.
The second item can be written formally as
\begin{equation*}
  -\mfrac{1}{NT}\log\PP^N\big(\{\overline Q^N(T) \approx \overline q\}\big) \xrightarrow[N\to\infty]{} \mfrac1T \Psi_{T,m}(\overline q\mid\rho^*) \xrightarrow[T\to\infty]{} \Psi_m(\overline q).
\end{equation*}
 In order for the statement to be interpretable as a large-deviation principle, the second convergence must be a Gamma limit, which holds by Proposition~\ref{prop:well-posedness of Psi}\eqref{it:Gamma-convergence}.
\end{proof}

\section{Dynamical phase transitions I: absence and occurrence}
\label{sec:convexity}

We now consider a system that consists of the graph $(\X,E)$, jump rates $k$ and total mass $m$ as explained in the previous section. In this section we study dynamical phase transitions for the general function $\Psi_m$, and study when phase transitions can~\emph{not} occur, and when they do occur. We then illustrate the conditions for a simple two-state example.

Let us first provide a more precise definition of dynamical phase transitions, following the general framework of \cite{BDSGJLL2005CurFluct,MR2227084,BDSGJLL2007LdpEmpCur,BDSGJLL2015MFT}:
\begin{definition}
We say that the system undergoes a dynamical phase transition in the average flow $\overline q\in\M(E)$ if
\begin{equation*}
  \Psi_m\left(\overline q\right) < \Psi^+_m\left(\overline q\right),
\end{equation*} 
where
\begin{equation}
  \Psi^+_m\left(\overline q\right) =
  \begin{cases}
    \inf_{\rho\in \M_m(\X)}\sum_{(x,y)\in E}s\left(\overline q_{x,y}\mid k_{x,y}(\rho)\right),  & \div\overline{q}=0,\\
    \infty,                                                                                     &\text{otherwise}.
  \end{cases}
\label{eq:defft}
\end{equation}
\label{def:DFT}
\end{definition}
Recall that $\Psi_m(\overline q)$ is the large-deviation cost to observe an atypical average flow $\overline q$. If the system does not undergo a dynamical phase transition for this $\overline q$, then the infimum over $\A_{T,m}$ can be taken over paths that are constant in time, which corresponds to an (atypical) stationary steady state. On the other hand, if a dynamical phase transition occurs at $\overline q$, then the system will have time-dependent behaviour when constrained to have an atypical flow $\overline q$. Furthermore, recall from~Proposition~\ref{prop:well-posedness of Psi} that when $\div\overline{q}\neq0$ then $\Psi_m(\overline q)=\infty$ and hence there can not be a dynamical phase transition.

\subsection{Dynamical phase transitions and lower and upper bounds}
\label{subsec:rule out}

We now study upper and lower bounds for $\Psi_m$, from which we deduce a sufficient condition under which dynamical phase transition can not occur, and a sufficient condition under which it does.
The upper bound $\Psi^+_m$ is already introduced in~\eqref{eq:defft}. A lower bound is constructed as follows. Considering the rate $k$ as an operator $\M_m(\X)\to \mathbb R_+^E$ with range
\begin{equation}
  \Ran_m(k):=\{(k_{x,y}(\rho))_{(x,y)\in E}\in \RR_+^E : \rho\in\M_m(\X) \},
\label{eq:Ran}
\end{equation}
one can rewrite the upper bound as
\begin{equation*}
  \Psi^+_m(\overline q)=\inf_{K\in \Ran_m(k)} \sum_{(x,y)\in E} s(\overline q_{x,y} \mid K_{x,y})\,.
\label{eq:equipsi+}
\end{equation*}
For the lower bound we introduce, denoting by $\co$ the convex hull,
\begin{equation*}
\label{eq:equipsi-}
  \Psi^-_m(\overline q):=\inf_{K\in \textrm{co}(\Ran_m(k))} \sum_{(x,y)\in E}s(\overline q_{x,y}\mid K_{x,y})\,.
\end{equation*}

\begin{proposition} For any divergence-free $\overline q$,
	\begin{equation*}
  	\Psi^-_m(\overline q)\leq \Psi_m(\overline q)\leq \Psi^+_m(\overline q)\,.
	\end{equation*}
	\label{prop:Psi bounds}
\end{proposition}

\begin{proof}
	Take an arbitrary $\rho\in\M_m(\X)$. Since $\overline q$ is divergence-free, the constant path $(\rho,\overline q)$ is an element of $\A_{T,m}(\overline q)$. Therefore
	\begin{align*}
  	\Psi_m(\overline q) \leq \lim_{T\to\infty} \mfrac{1}{T}\sum_{(x,y)\in E}\int_0^T\!s\big(\overline q_{x,y} \mid k_{x,y}(\rho)\big)\,dt = \sum_{(x,y)\in E} s\big(\overline q_{x,y} \mid k_{x,y}(\rho)\big).
	\end{align*}
	Optimising over all possible $\rho\in\M_m(\X)$ yields the upper bound.
	
  Now take an arbitrary path $(\rho,q)\in \mathcal A_{T,m}(\overline q)$. Applying Jensen's inequality to the jointly convex function $s(\cdot\mid\cdot)$ yields
	\begin{equation*}
	\sum_{(x,y)\in E}\frac 1T\int_0^T\!s\big(q_{x,y}(t)\mid k_{x,y}(\rho(t)\big)\,dt\geq \sum_{(x,y)\in E} s\Big(\overline q_{x,y}\mid {\textstyle \frac 1T\int_0^T\!k_{x,y}(\rho(t))\,dt} \Big)\geq \Psi^-_m(\overline q)\,,
	\end{equation*}
	because $T^{-1}\int_0^T\!k(\rho(t))\,dt \in \co(\Ran_m(k))$. The claimed lower bound is obtained by taking the infimum over all paths $(\rho(t), q(t))_{t\in\lbrack0,T\rbrack}\in \mathcal A_{T,m}(\overline q)$ and then letting $T\to\infty$.
\end{proof}


As a direct consequence we obtain the first condition to rule out phase transitions.
\begin{corollary}
If the set $\Ran_m(k)$ is convex then there are no dynamical phase transitions for any value of $\overline q$.
\label{cor:convex Ran}
\end{corollary}
\begin{proof}
  If $\Ran_m(k)$ is convex then $\Psi_m^-(\overline q)=\Psi_m(\overline q)=\Psi_m^+(\overline q)$.
\end{proof}

\begin{remark}
In this paper we are restricting to conservative models for which the total mass is preserved. The arguments are however easily extendible to non conservative models. In this case $\Ran(k)$ may be unbounded and the validity for example of condition of Corollary \ref{cor:convex Ran} are more likely to be verified.
This seems to confirm the phenomenon observed in \cite{PhysRevLett.102.250601} for which the addition of boundary conditions non preserving mass may inhibit dynamic phase transitions.
\end{remark}

The upper bound in Proposition \ref{prop:Psi bounds} can be improved, and the lower bound can be replaced by another one. To describe the lower bound, recall from~\eqref{eq:defft} that $\Psi_m^+(\overline{q})=+\infty$ if $\div\overline{q}\neq 0$. By contrast, we define another function
\begin{equation*}
  \Psi_m^\#(\overline{q}):=\inf_{\rho\in \M_m(\X)}\sum_{(x,y)\in E}s\left(\overline q_{x,y}\mid k_{x,y}(\rho)\right),
\end{equation*}
defined even in the case $\div \overline q\neq 0$, so that that $\Psi_m^\#(\overline q)=\Psi_m^+(\overline q)$ when $\div \overline q=0$. In the following lower and upper bounds ``$\co$'' now denotes the convex envelope.

\begin{proposition}\label{disturba}
For any $\overline{q}$,
\begin{equation*}
  \left(\co\Psi_m^\#\right)(\overline{q})\leq\Psi_m(\overline{q}) \leq \left(\co\Psi_m^+\right)(\overline{q})\,.
\end{equation*}
\end{proposition}
\begin{proof}
The upper bound follows by the upper bound in \ref{prop:Psi bounds} and the convexity of $\Psi_m$ (property $(ii)$ of Proposition~\ref{prop:well-posedness of Psi}).

The lower bound is obtained by the following argument. For any $T>0$ and any pair $(\rho,q)\in \A_{T,m}(\upbar q)$,
\begin{align*}
\frac 1T \mathcal I_{[0,T]}(\rho,q)&=\frac 1T\int_0^T\!\sum_{(x,y)\in E} s\big(q_{x,y}(t)\mid k_{x,y}(\rho(t))\big)\,dt\\
&\geq \frac 1T\int_0^T\!\Psi^\#_m(q(t))\,dt \geq \frac 1T\int_0^T\!\left(\co\Psi^\#_m\right)(q(t))\,dt\stackrel{\text{(Jensen)}}{\geq} \left(\co\Psi^\#_m\right)(\overline{q})\,.
\end{align*}
The result follows by minimising over $(\rho,q)\in \A_{T,m}(\upbar q)$ and taking $T\to\infty$.
\end{proof}

Together with the theory from Section~\ref{subsec:ldp time-average} we obtain a condition under which a phase transition occurs. 
\begin{corollary}
If $\Psi^+_m$ is not convex, then the system undergoes a dynamical phase transition.
\label{cor:Psiplus not convex}
\end{corollary}
\begin{proof}
In this case there exists a $\overline q\in \mathcal M(E)$ such that we have for the convex envelope $(\co\Psi^+_m)(\overline{q})<\Psi^+_m(\overline{q})$. By the upper bound in Proposition~\ref{disturba} we deduce that $\Psi_m(\overline{q})<\Psi^+_m(\overline{q})$ so that we have a dynamic phase transition at $\overline q$.
\end{proof}
When $\Psi_m^+$ is not convex and $\Psi_m=\co\Psi_m^+$ the typical behaviour of the system is to spend fractions of times with different typical densities and flows (see \cite{BDSGJLL2005CurFluct,MR2227084} for an extended discussion).

\subsection{Dynamical phase transitions and joint convexity}
\label{subsec:joint convexity}

In this section we derive an alternative condition to rule out dynamical phase transitions, again by a convexity argument, but now by joint convexity of the function,
\begin{equation}
  \M_m(\X)\times\M(E) \ni (\rho,q) \mapsto \sum_{(x,y)\in E} s(q_{x,y} \mid k_{x,y}(\rho)) \in \RR_+\,.
\label{eq:sk}
\end{equation}

\begin{proposition}
If the jump rates $k_{x,y}$ are such that the function~\eqref{eq:sk} is jointly convex in $(\rho,q)\in \M_m(\X)\times\M(E)$, then there are no dynamical phase transitions.
	\label{prop:jointly convex}
\end{proposition}
\begin{proof}
  Consider a pair $(\rho,q)\in \A_{T,m}(\upbar q)$. By Jensen's inequality:
	\begin{align*}
	\frac 1T \mathcal I_{[0,T]}(\rho,q)&=\frac 1T\int_0^T\!\sum_{(x,y)\in E} s\big(q_{x,y}(t)\mid k_{x,y}(\rho(t))\big)\,dt\\
	&\geq \sum_{(x,y)\in E} s\!\left({\textstyle \frac 1T\int_0^T\! q_{x,y}(t)\,dt} \mid {\textstyle k_{x,y}\big(\frac 1T \int_0^T\!\rho(t)\,dt\big)}\right)\\
	&=:\sum_{(x,y)\in E} s\left(\overline q_{x,y}\mid k_{x,y}(\tilde \rho_T)\right)\geq \Psi^+_m(\upbar q)\,,
	\end{align*}
	where we used the definition of $\mathcal A_{T,m}(\overline q)$ and we called  $\tilde \rho_T:=\frac 1T\int_0^T\rho_s ds\in \M_m(\X)$. Taking the infimum over all $(\rho, q)\in \A_{T,m}(\overline q)$ and then letting $T\to\infty$ yields $\Psi_m(\overline q) \geq \Psi^+_m(\overline q)$, which shows that there are no dynamical phase transitions.
\end{proof}
As we will illustrate in Section~\ref{subsec:2d example} the above proposition is not very powerful, but can be strengthened by requiring joint convexity in a small neighbourhood only.

\begin{remark}
In~\cite{BodineauDerrida2005}, the authors construct a dynamical phase transition by considering small perturbations around time independent trials $(\tilde\rho, \bar q)$. We consider $(\rho_\epsilon(t), q_\epsilon(t))$ where $q_\epsilon(t)= \bar q + \epsilon v(t)$, for some periodic perturbation $v$ such that $\int_0^Tv(s)ds=0$ and $\rho_\epsilon$ is determined by the continuity equation. The first variation $\tfrac{d}{d\epsilon}\I_{(0,T)}(\rho_\epsilon, q_\epsilon)\rvert_{\epsilon=0}=0$ since $(\tilde\rho,\bar q)$ is a critical point for the time independent trials. If one can find a $v(t)$ for which the second variation becomes negative, this shows that a small periodic perturbation can decrease the cost of a constant state, which corresponds to a dynamical phase transition. For a general local cost function $\L$ (as in the introduction), the second variation is
\begin{equation*}
  \tfrac{d^2}{d\epsilon^2}\I_{(0,T)}(\rho_\epsilon, q_\epsilon)\rvert_{\epsilon=0} = 
    \int_0^T\! \begin{bmatrix} m(t)\\v(t)\end{bmatrix}^\mathsf{T} \nabla^2\L(\tilde \rho, \overline q) \begin{bmatrix} m(t)\\ v(t)\end{bmatrix} \, dt.
\end{equation*}
where at first order in $\epsilon$ we have $\rho_\epsilon(t)=\tilde\rho+\epsilon m(t)$.
Therefore, the second variation cannot become negative if $\L$ is jointly convex, which is related to our Proposition~\ref{prop:jointly convex}.
\label{rem:second variation}
\end{remark}

\subsection{A simple two-state example}
\label{subsec:2d example}

We give a geometric illustration of the conditions from the previous section in the  simplest network with $|\X|=2$, i.e. $\X=\{1,2\}$. In this case $\M_m(\X)=\{(\rho_1,m-\rho_1):\rho_1\in\lbrack0,m\rbrack\}$ is one dimensional, and any choice of the jump rates can be rewritten as a zero-range process; this means that the rate at which the mass flows exiting from a given site only depends on the amount of mass present at that site. By Proposition~\ref{prop:well-posedness of Psi} we only need to consider divergence-free flows $\div \bar q=0$; this corresponds to $\bar q_{1,2}=\bar q_{2,1}$ both of which we shall denote as $\bar q$. For simplicity (and with a slight abuse of notation) we consider an homogeneous model, i.e. we assume that $k_{1,2}(\rho)=k(\rho_1)$ and $k_{2,1}(\rho)=k(m-\rho_1)$ for some function $k:\RR_+\to\RR_+$. 

We first study the implication of the range condition from Corollary~\ref{cor:convex Ran} for the two-state example. The range~\eqref{eq:Ran} can now be parametrised by
\begin{equation*}
  \Ran_m(k):= \left\{\Big(\alpha, k(m-k^{-1}(\alpha))\Big) :\alpha\in\lbrack 0,k(m)\rbrack\right\}.
\end{equation*}
Clearly for a linear $k$ the function $\alpha\to k(m-k^{-1}(\alpha))$ is linear and hence its graph is a convex set, see Figure~\ref{fig:phase space}(A), which rules out dynamical phase transitions by Corollary~\ref{cor:convex Ran}. This special case corresponds to independent particles jumping on the network. If $k$ is a convex  function, then $\alpha\to k(m-k^{-1}(\alpha))$ is a convex function, in which case its graph $\Ran_m(k)$ is not a convex set, see Figure~\ref{fig:phase space}(B). However, for a small $\overline q$ the gradient of $K\mapsto\sum_{(x,y)\in E}s(\overline q_{x,y}\mid K_{x,y})$ points left/downwards, and so the minimiser $K\in\co(\Ran_m(k))$ in $\Psi^-(\overline q)$ will still lie on $\Ran(k)$, implying that $\Psi^-(\overline q)=\Psi^+(\overline q)$. So for a convex function $k$ dynamical phase transitions can only occur for flows $\overline q$ that lie above the set $\Ran_m(k)$, see Figure~\ref{fig:phase space}(B). Similarly, for a concave function $k$, dynamical phase transitions are only possible for flows that lie below the set $\Ran_m(k)$, see Figure~\ref{fig:phase space}(C).

\begin{figure}[h!]
  \centering
  \subfloat[$k$ linear]{%
    \begin{tikzpicture}[scale=0.75]
      \tikzstyle{every node}=[font=\scriptsize];
      \draw[->](0,0)--(0,4) node[anchor=east]{$k_{2,1}$};
      \draw[->](0,0)--(4,0) node[anchor=north]{$k_{1,2}$};
      \draw(0,3) node[anchor=east]{$k(m)$} --node[above,pos=0.2,sloped]{$\Ran_m(k)$} (3,0) node[anchor=north]{$k(m)$};
      \draw[->](0,0)--(3,3) node[anchor=south west]{$\overline q$};
    \end{tikzpicture}
  }
  \subfloat[$k$ convex]{%
    \begin{tikzpicture}[scale=0.75]
      \tikzstyle{every node}=[font=\scriptsize];
      \draw[->](0,0)--(0,4) node[anchor=east]{$k_{2,1}$};
      \draw[->](0,0)--(4,0) node[anchor=north]{$k_{1,2}$};
      \filldraw[lightgray](0,3) .. controls (0,1.7) and (1.7,0) .. (3,0)--(0,3);
      \draw(0,3) node[anchor=east]{$k(m)$} .. controls (0,1.7) and (1.7,0) .. (3,0) node[below=-1,pos=0.3,sloped]{$\Ran_m(k)$} node[anchor=north]{$k(m)$};
      \draw[dotted](0,3)--(3,0);
      \draw(0,0)--(1.02,1.02);      
      \draw[double](1.02,1.02)--(2.97,2.97);\draw[->](2.99,2.99)--(3,3) node[anchor=south west]{$\overline q$};
    \end{tikzpicture}
  }
  \subfloat[$k$ concave]{%
    \begin{tikzpicture}[scale=0.75]
      \tikzstyle{every node}=[font=\scriptsize];
      \draw[->](0,0)--(0,4) node[anchor=east]{$k_{2,1}$};
      \draw[->](0,0)--(4,0) node[anchor=north]{$k_{1,2}$};

      \filldraw[lightgray](0,3) .. controls (1.3,3) and (3,1.3) .. (3,0)--(0,3);
      \draw(0,3) node[anchor=east]{$k(m)$} .. controls (1.3,3) and (3,1.3) .. (3,0) node[above,pos=0.3,sloped]{$\Ran_m(k)$} node[anchor=north]{$k(m)$};
      \draw[dotted](0,3)--(3,0);
      \draw[double](0,0) --(1.98,1.98);
      \draw[->](1.98,1.98)--(3,3) node[anchor=south west]{$\overline q$};
    \end{tikzpicture}
  }
  \caption{Three possible phase space plots. The grey areas depict the convex hull $\co(\Ran_m(k))$, and double stroked lines depict divergence-free flows $\overline q$ for which dynamical phase transition are not ruled out by Corollary~\ref{cor:convex Ran}.} 
  \label{fig:phase space}
\end{figure}
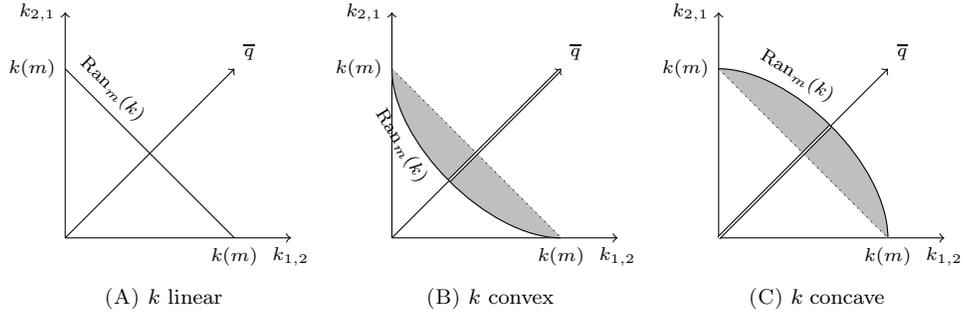

\begin{remark} From Figure~\ref{fig:phase space} one sees that by changing the total mass $m$ in the system the set $\Ran_m(k)$ can be moved upward or downward, which influences the region of possible dynamical phase transitions. The flexibility to play with $m$ will become more apparant and crucial in the next section where we construct an example of a dynamical phase transition. 
\end{remark}

We now study what the joint convexity condition of Proposition~\ref{prop:jointly convex} implies for the two-state homogeneous example. In that case joint convexity of the sum~\eqref{eq:sk} is implied by the joint convexity of the function $\RR^2\ni(\rho,q) \mapsto s(q\mid k(\rho))$, which corresponds to positive semidefiniteness of its Hessian matrix:
\begin{equation*}
	\begin{bmatrix}
	- q(\log k)''(\rho) + k''(\rho)  & -(\log k)'(\rho) \\
	-(\log k)'(\rho)            & 1/q
\end{bmatrix}.
\end{equation*}
By Sylvester's criterion, this matrix is positive semidefinite if and only if all its principal minors are non-negative, i.e.:
$$
\left\{
\begin{array}{l}
- q (\log k)''(\rho) + k''(\rho)\geq 0,\\ 
1/q \geq 0 ,\\ 
- (\log k)''(\rho) + k''(\rho)/q - \left((\log k)'(\rho)\right)^2 \geq 0\,.
\end{array}
\right.$$
The second inequality is trivial since $q\geq 0$, and the first inequality is contained in the third. Proposition~\ref{prop:jointly convex} thus means that dynamical phase transitions do not occur when the third inequality holds, that is:
\begin{equation}
	k''(\rho)\left(\frac1q- \frac1{k(\rho)}\right) \geq 0.
\label{eq:jointly convex condition}
\end{equation}
Actually this inequality holds for \emph{all} $\rho,q$ if and only if $k''(\rho)\geq0$ and $-k''(\rho)/k(\rho)\geq0$, which can only be true for linear functions $k$. In this sense Proposition~\ref{prop:jointly convex} is just ruling out the independent particles case. However, the same argument can be used to rule out phase transitions obtained by small oscillatory perturbations within a small neighbourhood of a critical time-independent $(\rho,\overline{q})$ (see Remark \ref{rem:second variation}). With this interpretation, we can deduce from~\eqref{eq:jointly convex condition} that for a convex $k$ dynamical phase transitions (with small oscillations) can only occur for large flows $\overline q$, and the other way around for concave $k$. Indeed, this statement is consistent with the phase space diagrams from Figure~\ref{fig:phase space}.

An exact characterisation of when a dynamical phase transition indeed occurs is computationally not easy, even in the simple two-state case. Instead of using a direct computational approach we will show the existence of a dynamical phase transition using a limiting procedure and a comparison with a continuous framework where the computations are easier.

\section{Dynamical phase transition II: an example}
\label{sec:discrete ring}

In the previous section we mostly focussed on sufficient conditions under which a dynamical phase transition can not occur; in this section we present an example where it does. The dynamics corresponds to a totally asymmetric zero range dynamics on a discrete ring with $L$ sites. The dynamics is spatially homogeneous and with strictly increasing jump rates
\begin{equation}
  k_{x,y}(\rho):=L\kappa(\rho_x)\mathds1_{x+1}(y),
\label{eq:k on ring}
\end{equation}
that we discuss in more detail in Section~\ref{subsec:discrete ring}, see also Figure~\ref{fig:directed ring}. Note that the rates $k$ on the left hand side of \eqref{eq:k on ring} depend on the size $L$ even if we are not writing explicitely such a dependence. We will fix the function $\kappa$ on the right hand side of \eqref{eq:k on ring} in such a way that  \eqref{eq:sk} is strictly concave on a region and such that there is just one single globally attractive stationary solution to equation \eqref{eq:closedrho}, which is natural in light of Remark~\ref{rem:single ss}.
We refer to \cite{MR3619540, MR3862457}
for discussions of dynamical phase transitions for asymmetric dynamics of particle systems in the hydrodynamic rescaling. We stress again that the limit considered here is very different from the hydrodynamic one.

More precisely we have the following result that allows to deduce the existence of a dynamical phase transition.

\begin{theorem} Consider a family of models on rings with $L$ sites and defined by the rates \eqref{eq:k on ring}. Assume that there exist $\overline\Phi,\tau,\delta>0$ and $\mu\geq0$ such that the function
\begin{equation}
  \phi \mapsto s\big(\phi \mid \kappa(\mu+\tau\phi)\big)
\label{eq:concavity}
\end{equation}
is strictly concave on $\lbrack \overline\Phi-\delta,\overline\Phi+\delta\rbrack$. If for average flow $\overline{q}_{x,x+1}\equiv \overline\Phi$ and any large $L>0$, the minimisation problem in the definition~\eqref{eq:defft} of $\Psi^+(\overline{q})$ has a unique minimiser, then for $L$ sufficiently large and $m:=L(\mu+\tau\overline\Phi)$ the system with rates $L\kappa$ undergoes a dynamical phase transition in $\overline{q}$.
\label{th:DPT}
\end{theorem}
The dynamical phase transition will be constructed in two stages. In Section~\ref{cont-comp} we perform an auxiliary computation on a continuous ring. We introduce a cost functional that is a continuous version of the discrete one and show that for this functional a travelling wave is less costly than a constant flow. In Section~\ref{subsec:approximation} we discretize the continuous travelling wave to show that the same result holds on the discrete ring, for $L$ sufficiently large but finite. This will conclude the proof of the result. In Section~\ref{subsec:examples} we present an example of a rate $k$ and flow $\overline{q}$ that satisfy the conditions in Theorem \ref{th:DPT} so that a dynamical phase transition takes place.

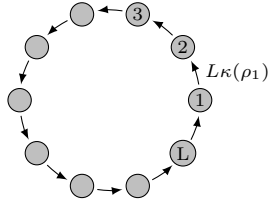
\begin{figure}[h!]  
\centering
\begin{tikzpicture}[scale=0.4]
\tikzstyle{every node}=[font=\scriptsize]
\def \n {10}
\def \m {9}
\def \radius {3}
\def \margin {10} 

  \node[draw, fill=lightgray,circle,inner sep=1] at ({360/\n * (1 - 1)}:\radius){1};
  \draw[->, >=latex] ({360/\n * (1 - 1)+\margin}:\radius) 
    arc ({360/\n * (1 - 1)+\margin}:{360/\n * (1)-\margin}:\radius) node[midway,anchor=west]{$L\kappa(\rho_1)$};
\foreach \s in {2,...,3}
{
  \node[draw, fill=lightgray,circle,inner sep=1] at ({360/\n * (\s - 1)}:\radius){\s};
  \draw[->, >=latex] ({360/\n * (\s - 1)+\margin}:\radius) 
    arc ({360/\n * (\s - 1)+\margin}:{360/\n * (\s)-\margin}:\radius);
}
\foreach \s in {4,...,\m}
{
  \node[draw, fill=lightgray,circle,inner sep=1] at ({360/\n * (\s - 1)}:\radius){\phantom{1}};
  \draw[->, >=latex] ({360/\n * (\s - 1)+\margin}:\radius) 
    arc ({360/\n * (\s - 1)+\margin}:{360/\n * (\s)-\margin}:\radius);
}
  \node[draw, fill=lightgray,circle,inner sep=1] at ({360/\n * (\n - 1)}:\radius){L};
  \draw[->, >=latex] ({360/\n * (\n - 1)+\margin}:\radius) 
    arc ({360/\n * (\n - 1)+\margin}:{360/\n * (\n)-\margin}:\radius);
\end{tikzpicture}
\caption{A directed discrete ring}
\label{fig:directed ring}
\end{figure}

\subsection{Totally asymmetric spatially homogeneous zero-range processes on a discrete ring.} 
\label{subsec:discrete ring}

The graph now consists of a discrete ring $\X:=\{1,2,\dots ,L\}$ with directed nearest-neighbour edges $E=\{(x,x+1)\,, x=1,\dots , L\}$ where the sum is modulo $L$, see Figure~\ref{fig:directed ring}. This means that mass can flow through the directed ring in one direction only, and that the divergence-free condition from Proposition~\ref{prop:well-posedness of Psi} imposes that $\overline q$ is a constant.

We assume that the jump rates are zero-range and spatially homogeneous, i.e. \eqref{eq:k on ring} for some smooth strictly increasing function $\kappa: \mathbb R_+\to \mathbb R_+$ satisfying the assumptions mentioned in Section~\ref{subsec:particle system}. Indeed, this setting is similar to the two-state model from Section~\ref{subsec:2d example}, where the rates are now rescaled with a factor  $L$.

The limiting system of equations~\eqref{eq:closedrho} associated to this class is given by
\begin{equation}\label{idring}
  \dot\rho_x(t)=L\kappa(\rho_{x-1}(t))-L\kappa(\rho_x(t))\,, \qquad x\in \X\,.
\end{equation}
with corresponding large-deviation rate functional for the paths,
\begin{equation*}
  \I_{\lbrack0,T\rbrack}(\rho,q)=\sum_{x=1}^L \int_0^T\!s\big(q_{x,x+1}(t)\mid L\kappa(\rho_x(t))\big)\,dt.
\end{equation*}
Since $\kappa$ is strictly increasing, for any fixed $m$ there exists a unique stationary solution $\overline\rho\in \M_m(\X)$ of \eqref{idring}, given by $\overline\rho_x=\frac mL$. Similarly, if the minimiser in the Definition~\eqref{eq:defft} of $\Psi_m^+(\overline{q})$ is unique, then by symmetry,
\begin{equation}\label{psipiu}
  \Psi^+_m(\overline q)= L s\big(\overline q \mid L\kappa(m/L)\big)\,.
\end{equation}

As mentioned in Remark~\ref{rem:single ss}, for systems with multiple equilibria and unstable behaviours one can easily imagine the occurrence of dynamical phase transitions, especially if the limit equation~\eqref{idring} itself already has a time-dependent solution. Therefore it is important that the equilibrium $\overline\rho$ is unique and globally attractive, so that the typical behaviour of the system is to relax to this equilibrium. This follows automatically from the fact that the equation is dissipative, which we will now derive from the monotonicity of $\kappa$.

The invariant measure of the zero range model is known to be:
  $$\nu^N\left(\rho^N\right)=\prod_{x=1}^L\frac{1}{Z\widehat{L \kappa}(\rho^N_x)}\,,
  $$ 
when $\rho^N\in\M_m(\X)\cap (\tfrac{m}{N}\NN)^L$. Here $\widehat{L\kappa}\left(\frac {m i}{N}\right):=\prod_{j=1}^i L\kappa\left(\frac{m j}{N}\right)$ for $i\geq 1$ and $\widehat{L\kappa} (0):=1$, and $Z$ is a normalisation factor, see~\cite[Sec.~2.3]{Kipnis1999}. A direct computation shows that when distributed according to $\nu_N$ and conditioned on the total mass $\sum_x\rho^N_x=m$, then $\rho^N$ satisfies a large-deviation principle with speed $N$ in $\M_m(\X)$ with rate functional:
\begin{equation}\label{clacdir}
  V_m(\rho)=\frac 1m \sum_{x=1}^L\int_{\overline\rho}^{\rho_x}\!\log \kappa(\alpha)\,d\alpha\,,
\end{equation}
where the equilibrium $\overline\rho$ now serves as a normalisation constant such that $V_m(\overline\rho)=0$. Note that due to monotonicity of $\kappa$, the function $V_m$ is convex and has unique minimiser $\overline\rho$. By macroscopic fluctuation techniques~\cite[Prop.~5.3]{RengerZimmer2019TR} it can be proved that the functional $V_m$ decreases along solutions $\rho(t)$ of~\eqref{idring}, which can also be seen from the direct computation:
\begin{multline*}
  \frac{d V_m(\rho(t))}{dt}=\sum_{x\in L} \frac Lm\Big[\kappa(\rho_{x-i}(t))-\kappa(\rho_x(t))\Big]\log \kappa(\rho_x(t))\\=-\frac Lm\sum_{x\in L} s\Big(\kappa(\rho_x(t))\Big|\kappa(\rho_{x+1}(t)\Big)\leq 0\,.
\end{multline*}
This shows that $V_m$ is indeed a Lyapunov functional, from which we deduce the asymptotic stability, that is, the limit equation~\eqref{idring} is dissipative and so $\overline{\rho}$ is globally attractive.

\begin{remark}
For models for which the invariant measure cannot be computed explicitly, \eqref{clacdir} and similar expressions can often still be derived from a Hamilton-Jacobi equation along the lines of a discrete version of Macroscopic Fluctuation Theory.
\end{remark}

\subsection{An inequality on the continuous torus}
\label{cont-comp}

Here we introduce an auxiliary functional on smooth density and flows that is related to the rate functional for a large ring under suitable conditions. 
It may correspond heuristically to a model on a \emph{continuous} ring where mass can flow anticlockwise only. Although it is not obvious whether this framework actually corresponds to a large-deviation principle (flows may develop shocks), the computation will be useful to prove and to understand the origin of the dynamical phase transition. Indeed the scaling limit of the discrete functionals is an interesting problem. Quantities on the continuous ring will be denoted by a hat (\;$\hat{\,}$\;).

On the continuous torus, $[0,1]$ with periodic boundary conditions, we consider smooth non-negative flows and densities $\hat q(x,t),\hat\rho(x,t)\geq 0$ such that
the continuous version of the continuity equation is satisfied:
\begin{equation}\label{contc}
  \partial_t\hat \rho(x,t)+\partial_x \hat q(x,t)=0\,, \qquad x\in [0,1]\,.
\end{equation}
For any such pair the corresponding continuous version of the rate functional is given by
\begin{equation}\label{func}
  \hat\I_{[0,T]}(\hat q,\hat \rho):=\int_0^T\!\int_0^1\! s\big(\hat q(x,t) \mid \kappa(\hat\rho(x,t)\big)\,dx\,dt.
\end{equation}
Note that in this expression appears $\kappa$ and this functional will be approximated by discrete models with a large number $L$ of sites and having rates $L\kappa$.  

We first construct a specific travelling wave solution of equation~\eqref{contc}, under the assumption of Theorem~\ref{th:DPT} with given $\overline\Phi,\tau,\delta>0$ and $\mu\geq0$. Let $\Phi:[0,1]\to\lbrack\overline\Phi-\delta,\overline\Phi+\delta\rbrack$ be a smooth and periodic non-constant function with average $\int_0^1\!\Phi(x)\,dx=\overline\Phi$, and define the travelling wave ($\TW$) as:
\begin{align}
  \hat \rho^\TW(x,t):=\mu + \tau\Phi(x-t/\tau).
    &&\text{and}&&
  \hat q^\TW(x,t):=\Phi(x-t/\tau)
\label{eq:cont travelling wave}
\end{align}
Then clearly, this pair $(\hat q^\TW,\hat \rho^\TW)$ solves \eqref{contc}, it is $\tau$-periodic in $t$ and $1$-periodic in $x$. Without loss of generality we may assume that $T$ is a multiple of $\tau$, so that the average flow is
$$
  \frac 1T \int_0^T\!\hat q^\TW(x,t)\,dt=\int_0^1 \Phi(x) dx:=\overline \Phi,
$$
and the corresponding total mass is given by 
$$
  m:=\int_0^1\!\hat\rho^\TW(x,t)\,dx=\mu+\tau\overline\Phi\,.
$$
We may therefore write $(\hat\rho^\TW,\hat q^\TW)\in\M_m([0,1])\times\M([0,1])$. The functional~\eqref{func} for this special pair is
$$
  \hat\I_{[0,T]}(\hat\rho^\TW,\hat q^\TW)=\int_0^T\!\int_0^1\!s\big(\Phi(x-t/\tau) \mid \kappa( \mu + \tau\Phi(x-t/\tau))\big)\,dx\,dt.
$$
Since the function $\Phi$ is periodic of period one, exchanging the order of integrations, the integrand of the integral on $dx$ does not depend on $x$ and the integral can be removed computing the integrand on $0$. With the change of variable $u=-t/\tau$ and again assuming that $T$ is a multiple of $\tau$,
\begin{equation*}
  \frac1T\hat\I_{[0,T]}(\hat \rho^\TW,\hat q^\TW)=\int_0^1\!s\big(\Phi(u) \mid  \kappa( \mu + \tau\Phi(u))\big)\,du.
\end{equation*} 

Next, let us compare the travelling wave to its corresponding constant profile
\begin{align}
  \hat \rho^\const(x,t):=\mu + \tau\overline\Phi.
    &&\text{and}&&
  \hat q^\const(x,t):=\overline\Phi
\label{eq:cont const}
\end{align}
Then clearly $\int_0^1\!\hat\rho^\const(x,t)\,dx=m=\int_0^1\!\hat\rho^\TW(x,t)\,dx$ and $\tfrac1T\int_0^T\!\hat q^\const(x,t)\,dt=\overline\Phi=\tfrac1T\int_0^1\!\hat q^\TW(x,t)\,dt$. By the strict concavity of~\eqref{eq:concavity} and the fact that $\Phi$ is not constant, Jensen's inequality yields that
\begin{multline}
  \frac1T\hat\I_{[0,T]}(\hat \rho^\TW,\hat q^\TW) =\int_0^1\!s\big(\Phi(u) \mid  \kappa( \mu + \tau\Phi(u))\big)\,du \\
    < \!s\big(\overline \Phi \mid  \kappa( \mu + \tau\overline\Phi)\big) = \frac1T\hat\I_{[0,T]}(\hat\rho^\const,\hat q^\const),
\label{eq:cont DPT}
\end{multline}
which can be heuristically interpreted as a dynamical phase transition on the continuous torus. We will use this strict inequality to prove the existence of a dynamical phase transition for a finite ring.

\subsection{Discretisation of density-flow pairs} 
\label{subsec:approximation}

The construction of the dynamical phase transition on the finite ring is obtained by discretising any smooth density and flow pair $(\hat\rho,\hat q)\in\M_m([0,1])\times\M([0,1])$ as follows. We embed the ring on the continuous torus associating the vertex $x\in \left\{1,\dots ,L\right\}$ to the point $x/L\in [0.1]$. We consider discretized masses and flows
\begin{align*}
  \rho_x^L(t):=L\int_{\frac{x-1/2}L}^{\frac{x+1/2}L}\hat\rho(y,t)\,dy
    &&\text{and}&&
  q_{x,x+1}^L(t):= L\hat q\big(\tfrac{x+1/2}L,t\big).
\end{align*}
It is easy to see that if $(\hat\rho,\hat q)$ satisfies the continuous continuity equation~\eqref{contc} then $(\rho_x^L,q_x^L)$ satisfies the discrete continuity equation~\eqref{eq:discr cont eq}.



We apply this discretisation to the two different profiles from the previous section. For the travelling wave $(\hat\rho^\TW,\hat q^\TW)$ from~\eqref{eq:cont travelling wave} with discretisation $(\rho^{L,\TW},q^{L,\TW})$. Recalling that the rates are $L\kappa$, we compute,
\begin{align}
  \frac1T\I_{[0,T]}(\rho^{L,\TW}, q^{L,\TW}) &=
\sum_{x=1}^L\frac LT\int_0^T\! s\!\left( \textstyle \hat q^\TW\big(\tfrac{x+1/2}L,t) \,\Big\vert\, \kappa\Big( L \int_{\frac{x-1/2}L}^{\frac{x+1/2}L}\!\hat\rho^\TW(y,t)\,dy\Big)\right)dt \notag\\
&=
L^2 \int_0^1\! s\!\left(\textstyle\Phi(u) \,\Big\vert\, \kappa\Big(\mu + \tau L \int_{u-1/L}^u\!\Phi(z)\,dz\Big)\right)\,du.
\label{euna}
\end{align}
For the constant profile $(\hat\rho^\const,\hat q^\const)$ from~\eqref{eq:cont const} with discretisation~$(\rho^{L,\const},q^{L,\const})$,
\begin{equation}
  \frac1T\I_{[0,T]}(\rho^{L,\const}, q^{L,\const}) = L^2 s(\overline\Phi \mid \kappa\big(\mu+\tau\overline\Phi)\big) \stackrel{\eqref{psipiu}}{=} \Psi^+_{L(\mu+\tau\overline\Phi)}(L\overline\Phi).
\label{edue}
\end{equation}
Observe that both discretised paths $(\rho^{L,\TW},q^{L,\TW}),(\rho^{L,\const},q^{L,\const})$ are feasible in the sense that they lie in the set $\A_{T,L(\mu+\tau\overline\Phi)}(L\overline\Phi)$ from~\eqref{eq:defA}.

Without loss of generality we may assume that $\delta$ is sufficiently small so that $\Phi\geq\overline\Phi-\delta>0$ and so under the assumptions on $\kappa$, the function $\kappa\big(\mu + \tau L \int_{u-1/L}^u\!\Phi(z)\,dz\big)$ is bounded from above and from below away from zero. Therefore from formulas~\eqref{euna} and \eqref{edue}
\begin{align*} 
  &\lim_{L\to +\infty}\frac{1}{L^2T}\I_{[0,T]}(\rho^{L,\TW}, q^{L,\TW})=\int_0^1\!s\big(\Phi(u)\mid \kappa(\mu +\tau\Phi(u))\big)\,du, \quad\text{and}\\
  &\lim_{L\to +\infty}\frac{1}{L^2T}\I_{[0,T]}(\rho^{L,\const}, q^{L,\const})= s(\overline \Phi \mid \kappa(\mu+\tau\overline\Phi))\,.
\end{align*}
From the strict inequality \eqref{eq:cont DPT} we thus find that for sufficiently large but finite $L$, \eqref{euna} is strictly less than \eqref{edue}. This proves the claim of Theorem~\ref{th:DPT} that for this $L$ and $m:=L(\mu+\tau\overline\Phi)$ the system undergoes a dynamical phase transition for the average flow $\overline\Phi$.

\subsection{An example}
\label{subsec:examples}

For this example we take the classical Young function, see Figure~\ref{fig:Hplot}(A),
\begin{align*}
  \kappa(\rho)&:=(\rho+1)\log(\rho+1)-\rho,
\end{align*}
and $\mu=0$ and $\tau=1$.

Given $\overline q$, the second condition in Theorem~\ref{th:DPT} is easily checked by rewriting,
\begin{align}
  L s\big(\overline q\mid L\kappa(\rho)\big) = L \phi\log\mfrac{\overline q}{L} -L\overline q - L\overline q\log \kappa(\rho) + L^2 \kappa(\rho).
\label{eq:s expand}
\end{align}
Since $\kappa$ is strictly convex and strictly log-concave, this function is clearly strictly convex, which shows the uniqueness of the minimiser in the Definition of $\Psi^+(\overline q)$. Note that this strict convexity is related but a bit different from the global attractiveness of the limit equation that we derived in Subsection~\ref{subsec:discrete ring}, since in this case we force a possible non-typical flow $\overline{q}$ on the system.

To show the strict concavity of~\eqref{eq:concavity} we calculate the second derivative:
\begin{align}
  H(\phi)&:=\mfrac{d^2}{d\phi^2} s\big(\phi\mid \kappa(\phi)\big)\notag\\ 
    &= \frac{1}{\phi} + \frac{1}{\phi+1} - 2\frac{\log(\phi+1)}{(\phi+1)\log(\phi+1)-\phi}\notag\\
    &\qquad - \frac{\phi}{(\phi+1)\big((\phi+1)\log(\phi+1)-\phi\big)} + \frac{\phi(\log(\phi+1))^2}{\big((\phi+1)\log(\phi+1)-\phi\big)^2}
\label{eq:H}
\end{align}
It is easily computed that $H(1)<0$, and in fact $\lim_{\phi\downarrow0}H(\phi)=-\infty$, so we may choose any $\overline\Phi\leq 1$ and $\delta>0$ sufficiently small, see Figure~\ref{fig:Hplot}. It  follows from Theorem~\ref{th:DPT} that for any $\overline{q}\equiv\overline{\Phi}\leq 1$, there exists a sufficiently large $L$ so that a dynamical phase transition occurs
with mass $m=L\overline \Phi$ and rates $L\kappa$.

\begin{figure}[h!]
\captionsetup{margin=2pt}
\centering
  \subfloat[The Young function]{%
    \begin{tikzpicture}[scale=0.75]
      \tikzstyle{every node}=[font=\scriptsize];
      \draw[->](0,0)--(4,0) node[anchor=south]{$\rho$};
      \draw[->](0,0)--(0,4);
      \draw(0,0) .. controls (0.5,0) and (2,0.5)..(4,4) node[anchor=east]{$\kappa(\rho)$};
    \end{tikzpicture}
    \label{subfig:Young}
  }\qquad\qquad\qquad
  \subfloat[$H(\phi)$]{%
    \begin{tikzpicture}[scale=0.75]
      \tikzstyle{every node}=[font=\scriptsize];
      \draw[->](0,0) --(4,0) node[anchor=north]{$\phi$};
      \draw[->](0,0.2)--(0,-3.8);
      \draw(1,0.-0.07)--(1,0.07) node[anchor=south]{$1$};
      \draw(0.2,-3.8) .. controls (0.3,0.2) and (0.5,0.2).. (4,0.1);
      \draw[->](2,-1.8)--(1.2,-0.3);
      \draw[->](2,-1.9) node[anchor=west,align=left]{Dynamical\\ phase\\ transition}--(0.6,-1.3);
      \draw[->](2,-2)--(0.4,-3);
    \end{tikzpicture}
    \label{subfig:H}
  }
\caption{On the left, the Young function $\kappa(\rho)$, and on the right the second derivative of $s\big(\phi\mid \kappa(\phi)\big)$.}
\label{fig:Hplot}
\end{figure}
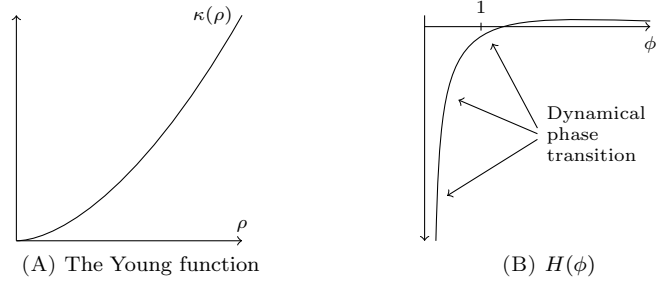


We mention that the choice of $\kappa$ for which a dynamical phase transition can occur is subtle. For example, a similar calculation as~\eqref{eq:H} can be done for the concave function 
$\kappa(\rho):= 2\log(\rho + 1)$, showing that $\lim_{\phi\downarrow0}H(\phi)=-1$. However, this function $\kappa(\rho)$ is strictly log-concave but not convex, and so \eqref{eq:s expand} ceases to be strictly convex for large $L$. Therefore, the minimiser in the definition of $\Psi_m^+(\overline{q})$ may no longer be unique, so possibly $\Psi_m^+(\overline{q})<Ls(\overline q\mid Lk(m/l))$, which breaks down the argument.

\section{Discussion}

We studied dynamical phase transitions on a finite graph, where the graph is fixed, and the number of particles and the time interval are simultaneously sent to infinity. For these models the action functional~\eqref{eq:LDP N} is entropic rather than quadratic. Therefore the dynamical phase transitions behave differently than previously studied models in the literature. We derived two different convexity conditions to rule out dynamical phase transitions: Corollary~\ref{cor:convex Ran} and Proposition~\ref{prop:jointly convex}. In particular the arguments that are worked out for the two-state example in Figure~\ref{fig:phase space} can be used to narrow down the search for possible dynamical phase transitions. As a byproduct of the bounds of Proposition~\ref{disturba} we also obtain a sufficient condition for a dynamical phase transition to occur, see Corollary~\ref{cor:Psiplus not convex}. Finally, we constructed an example of a dynamical phase transition for a specific zero-range process on a sufficiently large ring. It remains an open question whether dynamical phase transitions can occur on a graph with only two or three states. For smaller graphs, even in the case of two or three nodes, the exact structure of the rate function for the long-time averaged flow is more difficult to understand.

\appendix

\section{A Sanov Theorem for paths}
\label{app:ldp proof}

In this section we schetch  a simple argument for Theorem~\ref{th:discr path ldp} for the case of independent particles. A similar Sanov argument with interacting particles but without flows can be found in~\cite{Feng1994,Leonard1995}. The motivation of this Appendix is to give just the intuition on the specail form of the rate functional. An even more general proof can be found in~\cite{PattersonRenger2019} and \cite{Kraaij2017}.
\begin{proof}
Let $X_i(t), \, i=1,\hdots,N$ be independent copies of a Markov chain on the graph $(\X,E)$ with initial distribution $m^{-1}\rho^*$ and transition rates $k_{x,y}$. We consider the empirical measure of the random trajectories $X_{i}(\cdot):=\left(X_{i}(t)\right)_{t\in [0,T]}$ in $D([0,T];\X)$,
\begin{equation}\label{empp}
  \hat \PP^N:=\frac 1N \sum_{i=1}^N\delta_{ X_{i}(\cdot)}.
\end{equation}
By Sanov's Theorem, 
$\hat{\PP}^N$ satisfies a large-deviation principle in $\M_1(D([0,T];\X))$ with rate functional
\begin{equation}\label{re}
  \mathbb I_{[0,T]}\big(\hat{\PP}\big)=H\left(\hat{\PP} \mid \PP_{\rho^*}\right)\,,
\end{equation}
where $\PP_{\rho^*}$ is the law on $D([0,T];\X)$ of the Markov chain $X_1(\cdot)$ with transition rates $k_{x,y}$ and initial distribution
$m^{-1}\rho^*$, and $H(\cdot|\cdot)$ is the relative entropy. 

Let $\pi_t\lbrack x(\cdot)\rbrack := x(t)$ and $\Delta_t\lbrack x(\cdot)\rbrack:=(x(t^-),x(t))$ be the evaluation maps (assuming c\`adl\`ag paths). We can then write the pair $(\rho^N, W^N)$, defined in \eqref{eq:discr emp measure},\eqref{eq:discr emp flow}, as a function of $\hat{\PP}^N$:
\begin{equation}
\left\{
\begin{array}{l}
  \rho^N(t)=m\hat{\PP}^N\circ\pi_t^{-1}\,,\\
  Q^N(t)=m\hat{\PP}^N\circ\Delta_t^{-1}
\end{array}
\right.
\end{equation}
Hence by the contraction principle,
\begin{equation*}
  \I_{[0,T]}(\rho, q)=\inf_{\substack{\hat{\PP}\in\M_1(D([0,T];\X)):\\ \rho(\cdot)=m\hat{\PP}\circ\pi^{-1}_{(\cdot)},\, q(\cdot)=m\hat\PP\circ\Delta^{-1}_{(\cdot)} }} \mathbb{I}_{[0,T]}(\hat{\PP})\,,
\end{equation*}
It can be checked that the minimiser is a Markovian measure $\hat{\mathbb P}^{(\rho,q)}$ with time-dependent transition rates (assuming $\rho_x>0$),
\begin{equation}\label{correre}
  \hat k_{x,y}(t):=\frac{q_{x,y}(t)}{\rho_x(t)}\,,
\end{equation}
and initial condition $m^{-1}\rho^*$.

The relative entropy between two Markovian processes can be explicitly computed, and so
\begin{align*}
  \I_{[0,T]}(\rho, q)&=\mathbb{I}_{[0,T]}(\hat{\mathbb P}^{(\rho,q)})=H\!\left(\hat{\mathbb P}^{(\rho,q)}\Big |\mathbb P_{\rho^*}\right)\\
                     &=\sum_{(x,y)\in E}\int_0^T\!s\big(q_{x,y}(t) \mid \rho_x(t) k_{x,y} \big)dt\,.
\end{align*}
\end{proof}

\section*{acknowledgement}

This research has been partially funded by Deutsche Forschungsgemeinschaft (DFG) through grant CRC
1114 ``Scaling Cascades in Complex Systems'', Project Number 235221301, Project C08.

\bibliographystyle{plain}
\bibliography{library}

\end{document}